\documentclass[runningheads]{llncs}
\usepackage[T1]{fontenc}
\usepackage{latexsym}
\usepackage{amssymb}

\usepackage{amsmath}
\usepackage{amsthm}
\usepackage{booktabs}
\usepackage{enumitem}
\usepackage{graphicx}
\usepackage{xcolor}
\definecolor{pinegreen}{rgb}{0.0,0.47,0.44}
\usepackage{tabularx}
\usepackage{breqn}
\usepackage{algorithm}
\usepackage{algpseudocode}
\usepackage[mode=multiuser, layout=inline, status=final]{fixme}
\usepackage{appendix}
\fxsetup{theme=color}
\FXRegisterAuthor{bf}{ebl}{\color{pinegreen}Björn}
\FXRegisterAuthor{oe}{oel}{\color{blue}Özgür}
\FXRegisterAuthor{r}{ar}{\color{red}Raus oder kürzen: }

\begin{document}
\title{A Ratio-Based Shapley Value for Collaborative Machine Learning -- Extended Version}
\titlerunning{A Ratio-Based Shapley Value for Collaborative Machine Learning}
\author{Björn Filter\inst{1}\orcidID{0009-0008-8666-6239} \and Ralf Möller \inst{1}\orcidID{0000-0002-1174-3323} \and
Özgür Lütfü Özçep \inst{1}\orcidID{0000-0001-7140-2574} 
}
\authorrunning{B. Filter et al.}
\institute{Institute for Humanities-Centered AI (CHAI), University of Hamburg, Germany
\email{\{bjoern.filter, ralf.moeller, oezguer.oezcep\}@uni-hamburg.de}
}
\maketitle    
\begin{abstract}
Collaborative machine learning enables multiple data owners to jointly train models for improved predictive performance. However, ensuring incentive compatibility and fair contribution-based rewards remains a critical challenge. Prior work by Sim and colleagues \cite{sim2020collaborative} addressed this by allocating model rewards, which are non-monetary and freely replicable, based on the Shapley value of each party’s data contribution, measured via information gain. In this paper, we introduce a ratio-based Shapley value that replaces the standard additive formulation with a relative contribution measure. While our overall reward framework,including the incentive definitions and model-reward setting, remains aligned with that of Sim and colleagues, the underlying value function is fundamentally different. Our alternative valuation induces a different distribution of model rewards and offers a new lens through which to analyze incentive properties. We formally define the ratio-based value and prove that it satisfies the same set of incentive conditions as the additive formulation, including adapted versions of fairness, individual rationality, and stability. Like the original approach, our method faces the same fundamental trade-offs between these incentives. Our contribution is a mathematically grounded alternative to the additive Shapley framework, potentially better suited to contexts where proportionality among contributors is more meaningful than additive differences.

\end{abstract}

%%%%%%%%%%%%%%%%%%%%%%%%%%%%%%%%%%%%%%%%%%%%%%%%%%%%%%%%%%%%%%%%%%%%%%%%%%%

\section{Introduction}

Collaborative learning involves multiple participants contributing data to jointly train a machine learning model from which each participant can benefit. As artificial intelligence systems increasingly rely on multi-agent collaboration between institutions, individuals, or machines, ensuring fair and incentive-compatible mechanisms for cooperation becomes essential. However, this expansion simultaneously introduces complex challenges within AI safety \cite{hendrycks25AISafety}, specifically AI alignment \cite{russell22artificial} and cooperative AI \cite{conitzer23foundationsOFCooperativeAI,conitzer24socialChoice}. Even when the global goal is shared, the question of how to allocate rewards for individual contributions remains central to both the technical soundness and ethical viability of such systems.

Cooperative game theory provides powerful abstractions to address this challenge \cite{chalkiadakis2011computational}. A value function assigns a utility to each subset (coalition) of participants, representing the model’s quality trained on their combined data. A reward function then determines how the total value should be distributed among contributors. The most well-known example is the Shapley value \cite{shapley1953value}, which allocates rewards based on expected marginal additive contributions across all coalition orderings. This axiomatic approach has been widely applied to fair allocation problems across economics, AI, and beyond.

However, classical cooperative game theory assumes that rewards are indivisible and non-replicable \cite{wang2020principled,wang2023data}. In collaborative learning, the resource to be allocated are trained model or data sets, which are infinitely replicable. Any participant can receive a model copy at no additional cost. This fundamental difference calls for a rethinking of fairness and incentives. Recent work by Sim and colleagues \cite{sim2020collaborative} addresses this gap by introducing the notion of incentive-aware model rewards, where the reward is not a monetary share but a customized model which is scaled in predictive quality according to one’s contribution. Their framework adapts classical fairness axioms to this replicable setting and proves that a scaled Shapley value reward mechanism satisfies desirable incentive and fairness properties.

While the additive scaled Shapley value offers a principled starting point, it has limitations in collaborative learning scenarios. In particular, additive marginal contributions may not reflect the contextual significance of a participant's data. For instance, improving a weak model’s accuracy from 10\% to 20\% may be more impactful than boosting an already strong model from 90\% to 92\%. The additive approach rewards the latter more, despite its lower relative improvement.

This paper proposes a novel reward mechanism that replaces additive gains with multiplicative (ratio-based) contributions, capturing each participant’s relative impact on model performance. We define a ratio-based scaled Shapley value, which measures the proportional improvement a participant brings to coalitions and satisfies the same set of incentive and fairness axioms as the additive scaled Shapley value by Sim and colleagues. We then embed this valuation into the $\rho$-scaled reward scheme introduced by Sim and colleagues, enabling a smooth trade-off between social welfare maximization and strict fairness.

We prove that our ratio-based reward mechanism retains critical properties such as non-negativity, feasibility, individual rationality (under $\rho$-scaling), symmetry, and strict fairness criteria. Moreover, we show that it can offer more intuitive and equitable outcomes in data-diverse or redundancy-prone settings, where relative value carries more meaning than absolute marginal gain. A simulated comparison with the additive Shapley reward demonstrates the distinct behavior and potential advantages of our approach.

Our contribution extends the existing incentive-aware reward framework with a principled, ratio-based alternative that is better aligned with proportional fairness and contextual contribution. This offers new flexibility and interpretability in designing collaborative machine learning systems that are both fair and strategically robust.

%%%%%%%%%%%%%%%%%%%%%%%%%%%%%%%%%%%%%%%%%%%%%%%%%%%%%%%%%%%%%%%%%%%%%%%%%%%

\section{Problem Formulation}\label{prob}
We consider a setting in which a group of participants collaboratively train a machine learning model by contributing their individual data. Each participant $i \in N$ possesses a private dataset and faces the decision of whether to contribute it to a joint coalition for model training. In return for contributing, participants receive access to a model, whose predictive quality improves with the amount, diversity, and complementarity of the data aggregated.

The reward in this context is non-monetary: Access to a higher-quality shared model. Since access is replicable and costless to distribute, the key design challenge is to allocate model quality as a reward in a way that is fair and incentive-compatible, ensuring that contributors benefit according to their data’s value while also encouraging cooperation.

We model this as a cooperative game in characteristic form, in line with the modeling by Sim and colleagues \cite{sim2020collaborative}. Let  $N$ be the set of all players and $v: 2^N \rightarrow \mathbb{R}_{\geq 0}$ a  value function such that $v(\emptyset) = 0$. Each coalition $C \subseteq N$  corresponds to a subset of participants who train a model using their combined data, and the function $v(C)$ (also written $v_C$) gives the utility or quality of that model (e.g., accuracy, F1 score).

We assume that the value function $v$ is monotonic, i.e., for all $C' \subseteq C \subseteq N$, $v(C') \leq v(C)$), reflecting that additional data cannot reduce model quality. Depending on the application, $v$ may also be submodular or exhibit diminishing returns, due to redundancy or saturation effects.

The goal is to define a reward allocation rule $(r_i)_{i \in N}$ that assigns to each player $i$ a share of the model’s value, interpreted as their individual model reward. It describes, how much model access or quality a participant receives.

This formulation is particularly relevant to domains such as federated learning, cross-institutional medical AI, legal document classification, and historical manuscript digitization, where: Data is distributed among independent agents, model training benefits from aggregation, and contributors care about access to an improved model, not monetary profit.

%%%%%%%%%%%%%%%%%%%%%%%%%%%%%%%%%%%%%%%%%%%%%%%%%%%%%%%%%%%%%%%%%%%%%%%%

\section{Axioms for Incentivization and Fairness}\label{axioms}
An effective reward allocation mechanism in the described setting has to both incentivize players to participate and reward them fairly. To ensure this, Sim and colleagues have established a set of axioms that govern valid reward distributions \cite{sim2020collaborative}. Since we present an alternative reward scheme for the same setting, we will reiterate and briefly explain these axioms here. First, we give the incentive constraints that ensure feasibility, efficiency, and individual rationality. These constraints guarantee that the solution function remains viable and encourages participation. Building on this foundation, we state the fairness constraints to prevent unjust disparities in reward distribution. The following sections formalize these principles and their implications.

\subsection{Incentive Axioms}
To formulate a valid solution function, Sim and colleagues stated the following incentive constraints, ensuring that a solution function is feasible, efficient, and individually rational. These are common solution concepts from cooperative game theory as described in relevant textbooks \cite[Chapter 12]{shoham2008multiagent}, \cite{chalkiadakis2011computational}. For any $v \in \Gamma^N$ fulfilling the conditions stated in Section \ref{prob}, a valid solution function $\mathcal{M}(v)$ must fulfill the following axioms:\\
\\
\begin{tabularx}{\linewidth}{l X}
        R1 & \textbf{Non-negativity:} Each player must get a non-negative reward: $\forall i \in N$: $r_i \geq 0$.
\end{tabularx}
\begin{tabularx}{\linewidth}{l X}
        R2 &\textbf{Feasibility:} The reward for each player in any coalition $C \in CS$ cannot be larger than the value achieved by that coalition: $\forall C \in CS$, $\forall i \in C: r_i \leq v_C$.
\end{tabularx}
\begin{tabularx}{\linewidth}{l X}
        R3 & \textbf{Weak Efficiency:} In each coalition $C \in CS$, the reward received by at least one player $i \in C$ must be as large as the total value that coalition $C$ can achieve: $\forall C \in CS \exists i \in C: r_i = v_C$.
\end{tabularx}
\begin{tabularx}{\linewidth}{l X}
        R4 & \textbf{Individual Rationality:} Each player must receive a reward that is at least as large as the value that player can achieve by themselves: $\forall i \in N: r_i \geq v_i.$
\end{tabularx}\\
\\
R1 and R4 are the same axioms for solution concepts as in cooperative game theory with non-replicable rewards. R2 and R3 have been adapted from Chalkiadakis and colleagues \cite{chalkiadakis2011computational}, since in our setting we can give every member of a coalition $C\subseteq N$ a reward of up to $v_C$. We only require weak efficiency, since strong efficiency would imply having to pay out $v_C$ to every member of the coalition, which would maximize welfare, but impede any fairness considerations. Axiom R5 is necessary in our setting, since a solution function describes rewards for every player regardless of whether they are in the coalition. Thus, it has to be ensured that they are treated as if they work alone.

\subsection{Fairness Axioms}

To ensure a fair allocation of rewards, the solution function must meet the following four fairness axioms:\\

\noindent
\begin{tabularx}{\linewidth}{l X}
    F1 & \textbf{Uselessness:} If player $u$'s data does not provide an increase to the value of any coalition, then player $u$ should always receive a valueless reward. Furthermore, all other players in a coalition with $u$ should get identical rewards as they would if $u$ were not part of the coalition.  For all $i \in N$,
\end{tabularx}
\begin{align}
        \left( \forall C \subseteq N \setminus \{i\} \text{ with } C \neq \emptyset: v_C = v_{C \cup \{i\}} \right) \Rightarrow r_i = 0.
    \end{align}
\begin{tabularx}{\linewidth}{l X}
    F2 & \textbf{Symmetry:} If players $i$ and $j$ contribute identically to any coalition, they should receive equal rewards: For all $i,j \in N$ s.t. $i\neq j$,
\end{tabularx}
\begin{align}
        \left( \forall C \subseteq N \setminus \{i, j\}: v_{C \cup \{i\}} = v_{C \cup \{j\}} \right) \Rightarrow r_i = r_j.
    \end{align}
\begin{tabularx}{\linewidth}{l X}
    F3 & \textbf{Strict Desirability:} If the value of at least one coalition improves more by including player $i$ instead of player $j$, but the reverse is not true, player $i$ should receive a more valuable reward than $j$. For all $i,j \in N$ s.t. $i\neq j$, 
    {\begin{align}
        \begin{split}
            \left( \exists B \subseteq N \setminus \{i, j\}: v_{B \cup \{i\}} > v_{B \cup \{j\}} \right) & \land \\
            \left( \forall C \subseteq N \setminus \{i, j\} v_{C \cup \{i\}} \geq v_{C \cup \{j\}} \right) & \Rightarrow r_i > r_j.
        \end{split}
    \end{align}}
\end{tabularx}
\begin{tabularx}{\linewidth}{l X}
    F4 & \textbf{Strict Monotonicity:} Suppose the value of at least one coalition containing player $i$ improves (for example, by including more data of $i$), ceteris paribus, then player $i$ should receive a more valuable reward than before: Let $v$ and $v'$ denote any two value functions over all coalitions $C \subseteq N$ and $\mathcal{M}(v)^C_i$ and $\mathcal{M}(v')^C_i$ be the corresponding values of rewards received by player $i$ in coalition $C$. For all $C \in N$ and all $i \in C$,
\end{tabularx}
\begin{align}
    \begin{split}
        \left( \exists B \subseteq C \setminus \{i\}: v'_{B \cup \{i\}} > v_{B \cup \{i\}} \right) & \land\\
        \left(\forall C \subseteq N \setminus \{i\} v'_{C \cup \{i\}} \geq v_{C \cup \{i\}} \right) & \land\\
        \left( \forall A \subseteq N \setminus \{i\}: v'_A = v_A \right) \land \left( v'_N > r_i \right)  &\quad \Rightarrow   \quad r'_i > r_i.
    \end{split}
\end{align}

F1 and F2 are axioms of the Shapley value \cite{shapley1953value}. Axiom F3 was first introduced by Maschler and Peleg \cite{maschler1966characterization} and, in our setting, reduces to the fact that players who contribute larger values should receive larger rewards, compared to players with less valuable contributions. F4 is adapted from Young \cite{young1985monotonic}.

Together, these axioms were introduced by Sim and colleagues \cite{sim2020collaborative} as conditions for fairness. 
However, our axioms F3 and F4 are slightly weaker than theirs. In F3, we require $B \neq \emptyset$, that is, if two players $i$ and $j$ bring the same increase to every nonempty coalition, they may get the same reward in $C$, even if $v_i > v_j$. In F4, instead of $v'_C > v_C$, Sim and colleagues only required $\exists B \subseteq C \setminus \{i\}: v'_{B \cup \{i\}} > v_{B \cup \{i\}}$, i.e., if there is a set $B$ for which  $v'_{B \cup \{i\}} > v_{B \cup \{i\}}$, $i$ should get a higher reward not only in $B$ but also in coalitions containing $B$.

 For any game $v \in \Gamma^N$ that fulfills the conditions stated in Section \ref{prob}, for a function $\mathcal{M}(v)$ to be a valid solution, the following must hold:\\
\\
\begin{tabularx}{\linewidth}{l X}
    R5 & \textbf{Fairness:} For all $C \subseteq N$ and $i \in N$, the rewards $\mathcal{M}^C_i$ must satisfy F1 to F4.
\end{tabularx} \\

%%%%%%%%%%%%%%%%%%%%%%%%%%%%%%%%%%%%%%%%%%%%%%%%%%%%%%%%%%%%%%%%%%%%%%%%

\section{A Ratio-Based Shapley Value}

The original incentive-aware reward scheme by Sim and colleagues \cite{sim2020collaborative} is grounded in the Shapley value, which measures marginal contribution through subtraction. The marginal contribution of a player $i \in N$ to a coalition $C \subseteq N \setminus \{i\}$ is measured as
\begin{align}
    \Delta^{abs}_i =  v_{ C \cup \{i\} } - v _C.
\end{align}

It is the amount of value that is added to the coalition by $i$ joining. Using this, the Shapley Value evaluates a participant’s contribution as the expected marginal increase in utility (e.g., model performance) when they join a coalition. To get the average marginal contribution of a player $i \in N$, the players marginal contribution to all possible coalitions is summed up and divided by the number of possible coalition:

\begin{align}
    Shapley_v(i) = \frac{1}{n!} \sum_{\pi \in \Pi_N} \left( v_{S_{\pi, i} \cup \{i\}} - v_{S_{\pi, i}} \right),
\end{align}

where $\Pi_N$ is the set of all possible permutations of $N$ and $S_{\Pi, i}$ is the coalition of parties preceding $i$ in permutation $\pi$.

While this framework satisfies many desirable incentive properties such as fairness, symmetry, and efficiency, it implicitly assumes that absolute gains are the most meaningful measure of contribution.

However, in many collaborative machine learning settings, relative improvement may provide a more accurate reflection of a participant’s true impact. Consider, for instance, a participant who improves the performance of a weak model from 10 \% to 20 \% accuracy versus one who improves a strong model from 90 \% to 92 \%. The additive Shapley value would reward the latter more, due to a higher absolute gain. Yet the former has doubled the model’s performance, which is a compelling and arguably more meaningful contribution.

This observation motivated the development of our ratio-based valuation that captures a participant’s multiplicative influence on the performance of coalitions. Rather than asking "how much does this player increase utility?" we ask "by what factor does this player improve the coalition’s performance?" This shift from additive to multiplicative reasoning leads to a reward scheme that better reflects contextual importance, novelty, and leverage, especially in settings with heterogeneous or redundant data sources.

Formally, for any player $i \in N$ and any coalition $C \subseteq N \setminus \{i\}$, we define the relative marginal contribution of $i$ to $C$ as:

\begin{align}
    \Delta^{rel}_{i, C} :=
    \begin{cases}
        \frac{v_{ C \cup \{i\} }}{ v _C} - 1 & \text{ if } v_C \neq 0 \\
        0 & \text{ else}.
    \end{cases}
\end{align}

Using this, we can define the ratio-based Shapley value $\phi^{rel}_i$ of player $i$ as the expected relative improvement they generate, averaged over all possible coalitions:

\begin{align}\label{ratioShapley}
    \phi^{rel}_i := \frac{1}{n!}\sum_{\pi \in \Pi_N} \Delta^{rel}_{i, S_{\pi, i}}.
\end{align}

This parallels the classical Shapley value, but evaluates ratios rather than differences.

However, the distribution of rewards according to the Shapley value would not fulfill R3 since no agent would get the maximum reward. Hence, we then apply the same $\rho$-scaling used by Sim and colleagues to ensure that the agent with the largest contribution gets the maximum reward:

\begin{align}\label{ratioShapleyReward}
    r_i = \left(\frac{\phi^{rel}_i}{\phi^{\ast}_C} \right)^{\rho} \times v_C
\end{align}
where $\phi^{\ast}_C = \max_{i \in C}\phi^{rel}_i$ ensures normalization and $\rho \in [0, 1]$ controls the size of rewards, acting as a mediator between fairness and social welfare maximization.

This ratio-based formulation intuitively rewards participants who make coalitions better in proportion to their prior quality. This is especially relevant in collaborative settings with:

\begin{itemize}
    \item Heterogeneous data quality, where a small amount of high-quality data can drastically improve a weak model;
    \item Redundant contributions, where marginal additive gains become small due to overlapping information.
    \item Early-stage model building, where absolute gains are small but relative improvements are large.
\end{itemize}
Moreover, we show in the next sections that this formulation preserves key incentive properties: feasibility, non-negativity, individual rationality (with $\rho$-scaling), symmetry, and fairness axioms adapted for the ratio-based setting. This makes it a viable alternative within the same cooperative game framework, offering an alternative interpretation of "contribution" without sacrificing desirable theoretical guarantees.

\subsection{Ratio-based Scaled Shapley Value Fulfills Incentive and Fairness Axioms}

To be a viable alternative within cooperative machine learning, the ratio-based Shapley value must preserve the essential incentive properties that make the additive Shapley value attractive. In this subsection, we show that our formulation satisfies the set of rationality and fairness conditions adapted from the incentive-aware model reward framework introduced by Sim and colleagues \cite{sim2020collaborative} sated in section \ref{axioms}. Specifically, the following theorem summarizes the key properties satisfied by the ratio-based value defined above. We do not treat R4 here, since it is not necessarily fulfilled by any $\rho \in [0, 1]$ and therefore will be treated separately later on.  

\begin{theorem}
    The ratio-based Shapley value described in equation \ref{ratioShapleyReward} fulfills the  Axioms R1 to R3 as well as F1 to F4.
\end{theorem}

\begin{proof}
    Let $v \in \Gamma^N$ be monotonic and fulfill $v(\emptyset = 0)$ for all nonempty $C \subseteq N$. We will see that it fulfills Axioms R1 to R3 as well as F1 to F4.\\
    \\
    
    R1 \textbf{Non-negativity:}\\
    By definition, we have $v(C) \geq 0$ for all $C \subseteq N$.
    Furthermore, due to the monotonicity constraint, for all $C \subset N$ and all $i \notin C$, we have $v(C \cup \{i\}) \geq v(C)$, therefore $\frac{v(C \cup \{i\})}{v(C)} \geq 1$. Thus, for all $\pi \in \Pi_N$ with $v \left(S_{\pi, i} \right) > 0$,
    \begin{align}
        \frac{v\left( S_{\pi, i}\cup \{i\} \right)}{v \left(S_{\pi, i} \right)} - 1 \geq 0
    \end{align}
    has to hold. If $v \left(S_{\pi, i} \right) > 0$, by definition we have $\Delta^{rel}_{i, C} = 0$. Therefore, for all $i \in N$, $ \phi^{rel}_i \geq 0$ and thus $r_i \geq 0$.\\
    \\
    
    R2 \textbf{Feasibility:}\\
    By definition, for all $i \in N$ and all $\rho \in [0, 1]$ we have $\frac{\phi^{rel}_i}{\phi^{\ast}_C} \leq 1$ and therefore $r_i = \left(\frac{\phi^{rel}_i}{\phi^{\ast}_C} \right)^{\rho} \times v_C \leq 1^{\rho} \times v_C = v_C$.\\
    \\
    
    R3 \textbf{Weak Efficiency:}\\
    Let $i = \arg\max_{j \in C}\phi^{rel}_i$, so $\phi^{rel}_i = \phi^{\ast}_C$. Then by definition, we have $\frac{\phi^{rel}_i}{\phi^{\ast}_C} = 1$ and therefore $r_i = \left(\frac{\phi^{rel}_i}{\phi^{\ast}_C} \right)^{\rho} \times v_C = v_C$.\\
    \\
    F1 \textbf{Uselessness:}\\
    Suppose for some $i \in N$: $\forall C \subseteq N \setminus \{i\}$ with $C \neq \emptyset$: $v_{C \cup \{i\}} = v_C$, then
    \begin{align}
        \begin{split}
            \phi^{rel}_i & = \frac{1}{n!} \sum_{\pi \in \Pi_N}\Delta^{rel}_{i, S_{\pi, i}}\\
            & = \frac{1}{n!}\left( \sum_{\pi \in \Pi^+_N} \left( \frac{v\left( S_{\pi, i}\cup \{i\} \right)}{v \left(S_{\pi, i} \right)} - 1 \right) + \sum_{\pi \in \Pi^0_N} 0 \right)\\
            & = \frac{1}{n!} \sum_{\pi \in \Pi^0_N} \left( 1 - 1 \right) + \sum_{\pi \in \Pi^0_N} 0\\
            & = 0,
        \end{split}
    \end{align}
    where $\Pi^+_N$ denotes the subset of $\Pi_N$, such that for all $\pi \in \Pi^+_N$ $v_{S_{\pi, i}} > 0$ and $\Pi^0_N := \Pi_N \setminus \Pi^+_N$. Thus $r_i = \left(\frac{\phi^{rel}_i}{\phi^{\ast}_C} \right)^{\rho} \times v_C = \left(\frac{0}{\phi^{\ast}_C} \right)^{\rho} \times v_C = 0$.\\
    \\
    
    F2 \textbf{Symmetry:}\\
    Let $i, j \in N$ with $i \neq j$ and $\forall C \subseteq N \setminus \{i, j\}$ $v_{C \cup \{i\}} = v_{C \cup \{j\}}$. Let $\Pi^+_N$ and $\Pi^0_N$ be defined as above. For each $\pi \in \Pi_N$ we can obtain a corresponding $\pi'\in \Pi_N$ by swapping the position of $i$ and $j$. Consider the following two cases: (a) When $i$ precedes $j$ in $\pi$, $S_{\pi, i} = S_{\pi', j}$ and thus $v_{S_{\pi, i}} = v_{S_{\pi', j}}$. By setting $C = S_{\pi, i} = S_{\pi', j}$, we have $v_{S_{\pi, i} \cup \{i\}} = v_{C \cup \{i\}} = v_{C \cup \{j\}} = v_{S_{\pi', j} \cup \{j\}}$; (b) when $j$ precedes $i$ in $\pi$, $S_{\pi, i} \cup \{i\} = S_{\pi', j} \cup \{j\}$ and thus $v_{S_{\pi, i}\cup \{i\}} = v_{S_{\pi', j}\cup \{j\}}$. By setting $C = S_{\pi, i} \setminus \{j\} = S_{\pi', j} \setminus \{i\}$, we have $v_{S_{\pi, i} } = v_{C \cup \{j\}} = v_{C \cup \{i\}} = v_{S_{\pi', j}}$. Thus especially, $v \left(S_{\pi, i} \right) = 0$, iff $v \left(S_{\pi', j} \right) = 0$ Therefore:
    \begin{align}
        \begin{split}
            & \sum_{\pi \in \Pi_N} \Delta^{rel}_{i, S_{\pi, i}} = \sum_{\pi \in \Pi^+_N}\frac{v\left( S_{\pi, i}\cup \{i\} \right)}{v \left(S_{\pi, i} \right)} - 1\\
            = & \sum_{\pi' \in\Pi^+_N}\frac{v\left( S_{\pi', j}\cup \{j\} \right)}{v \left(S_{\pi', j} \right)} - 1 = \sum_{\pi' \in \Pi_N} \Delta^{rel}_{i, S_{\pi', i}}.
        \end{split}
    \end{align}
    Thus
    \begin{align}
        \phi_i^{rel} = \frac{1}{n!}\sum_{\pi \in \Pi_N} \Delta^{rel}_{i, S_{\pi, i}} = \frac{1}{n!}\sum_{\pi' \in \Pi_N} \Delta^{rel}_{j, S_{\pi, j}} = \phi_j^{rel}.
    \end{align}
    and therefore
    \begin{align}
        \begin{split}
            r_i & = \left(\frac{\phi^{rel}_i}{\phi^{\ast}_C} \right)^{\rho} \times v_C = \left(\frac{\phi^{rel}_j}{\phi^{\ast}_C} \right)^{\rho} \times v_C  = r_j.
        \end{split}
    \end{align}\\
    \\
    
    F3 \textbf{Strict Desirability:}\\
    Assume there are $i, j \in N$ with $i \neq j$ and $\forall C \subseteq N \setminus \{i, j\}$ $v_{C \cup \{i\}} \geq v_{C \cup \{j\}}$. Furthermore let there be some  $B \subseteq N \setminus \{i, j\}$ such that $v_{B \cup \{i\}} > v_{B \cup \{j\}}$. Let $\Pi^+_N$ and $\Pi^0_N$ be defined as above. For each $\pi \in \Pi_N$ we can obtain a corresponding $\pi'\in \Pi_N$ by swapping the position of $i$ and $j$. Similar to to the proof of F2, we show the inequality $r^C_i(v) > r_j^C(v)$ by considering the following two cases: (a) when $i$ precedes $j$ in $\pi$, $S_{\pi, i}= S_{\pi', j}$ (i.e., both excluding $i$ and $j$) and thus $v_{S_{\pi, i}} = v_{S_{\pi', j}}$. It also follows from the assumption that by setting $C = S_{\pi, i} = S_{\pi, i}$, $v_{S_{\pi, i}} = v_{C \cup \{i\}} \geq v_{C \cup \{j\}} = v_{S_{\pi', j}}$; (b) when $j$ precedes $i$ in $\pi$, $S_{\pi, i} \cup \{ i \} =S_{\pi', j} \cup \{ j \}$ (i.e., both containing $i$ and $j$) and thus $v_{S_{\pi ,i} \cup \{i\}} = v_{S_{\pi' ,j} \cup \{j\}}$. It also follows from our assumption that by setting $C = S_{\pi ,i} \setminus \{j\} = S_{\pi' ,j} \setminus \{i\}$, $v_{S_{\pi, i}} = v_{C \cup \{j\}} \leq v_{C \cup \{i\}} = v_{S_{\pi', j}}$. Due to our assumption, a strict inequality must hold for at least one $\pi, \pi'$ in either case a or b. Therefore:
        \begin{align}
        \begin{split}
            & \sum_{\pi \in \Pi_N} \Delta^{rel}_{i, S_{\pi, i}} = \sum_{\pi \in \Pi^+_N}\frac{v\left( S_{\pi, i}\cup \{i\} \right)}{v \left(S_{\pi, i} \right)} - 1\\
            > & \sum_{\pi' \in\Pi^+_N}\frac{v\left( S_{\pi', j}\cup \{j\} \right)}{v \left(S_{\pi', j} \right)} - 1 = \sum_{\pi' \in \Pi_N} \Delta^{rel}_{i, S_{\pi', i}}.
        \end{split}
    \end{align}
    Thus
    \begin{align}
        \phi_i^{rel} = \frac{1}{n!}\sum_{\pi \in \Pi_N} \Delta^{rel}_{i, S_{\pi, i}} > \frac{1}{n!}\sum_{\pi' \in \Pi_N} \Delta^{rel}_{j, S_{\pi, j}} = \phi_j^{rel}.
    \end{align}
    and finally
    \begin{align}
        \begin{split}
            r_i & = \left(\frac{\phi^{rel}_i}{\phi^{\ast}_C} \right)^{\rho} \times v_C > \left(\frac{\phi^{rel}_j}{\phi^{\ast}_C} \right)^{\rho} \times v_C  = r_j.
        \end{split}
    \end{align}\\
    \\
    
     F4 \textbf{Strict Monotonicity:}\\
     Let $\{v_C\}_{C \in 2^N}$ and $\{v'_C\}_{C \in 2^N}$ denote any two sets of values of data over all coalitions $C \subseteq N$, and $r_i$ and $r'_i$ be the corresponding values of model rewards received by party $i$. Let $\Pi^+_N$ and $\Pi^0_N$ be defined as above. For some $i \in N$, let
    \begin{align}\label{condF4}
        \begin{split}
            & \left( \exists B \subseteq C \setminus \{i\}: v'_{B \cup \{i\}} > v_{B \cup \{i\}} \right) \land
             \left(\forall C \subseteq N \setminus \{i\} v'_{C \cup \{i\}} \geq v_{C \cup \{i\}} \right) \land\\
             &\left( \forall A \subseteq N \setminus \{i\}: v'_A = v_A \right) \land \left( v'_N > r_i \right),
        \end{split}
    \end{align}
    then
    \begin{align}
        \begin{split}
            \phi^{rel}_i & = \sum_{\pi \in \Pi_N} \Delta^{rel}_{i, S_{\pi, i}} = \sum_{\pi \in \Pi^+_N}\frac{v\left( S_{\pi, i}\cup \{i\} \right)}{v \left(S_{\pi, i} \right)} - 1\\
            & < \sum_{\pi \in \Pi^+_N}\frac{v'\left( S_{\pi, i}\cup \{i\} \right)}{v' \left(S_{\pi, i} \right)} - 1 = \sum_{\pi \in \Pi_N} \Delta^{rel'}_{i, S_{\pi, i}} = \phi^{rel'}_i.
        \end{split}
    \end{align}
    Note that for any $\pi$, $S_{\pi, i}$ does not contain $i$. Hence, it follows that $v_{S_{\pi, i}} = v'_{S_{\pi, i}}$. However, there must be at least one $\pi$, such that $v_{S_{\pi, i} \cup \{i\}} < v'_{S_{\pi, i} \cup \{i\}}$.\\

    Now, we will show that the Geometric Shapley Value fulfills the \textbf{party monotonicity} property which will be used in the following. If $i$ fulfills \ref{condF4}, then for all $j \in N \setminus \{i\}$, $\phi'_i - \phi_i \geq \phi'j - \phi_j$:
    Let $v_C^{\Delta} := \left( v'_C - v_C \right)$ for any $C \subseteq N$. Then, for any $j \in N \setminus \{i\}$:
    \begin{align}
        \phi^{rel'}_j - \phi^{rel}_j & = \frac{1}{n!}\left( \sum_{\pi \in \Pi^+_{i \prec j}}\frac{v'_{S_{\pi, i}\cup \{j\}}}{v'_{S_{\pi, i}}} - \frac{v_{S_{\pi, i}\cup \{j\}}}{v_{S_{\pi, i}}} \right)\\
        & \leq \frac{1}{n!}\left( \sum_{\pi \in \Pi^+_{i \prec j}}\frac{v'_{S_{\pi, i}\cup \{j\}}}{v_{S_{\pi, i}}} - \frac{v_{S_{\pi, i}\cup \{j\}}}{v_{S_{\pi, i}}} \right)\\
        & = \frac{1}{n!}\left( \sum_{\pi \in \Pi_{i \prec j}}\frac{v^{\Delta}_{S_{\pi, i}\cup \{j\}}}{v_{S_{\pi, i}}} \right)\\
        & = \frac{1}{n!}\left( \sum_{\pi' \in \Pi_{j \prec i}}\frac{v^{\Delta}_{S_{\pi, i}\cup \{i\}}}{v_{S_{\pi, i}}} \right)\\
        & \leq \frac{1}{n!}\left( \sum_{\pi \in \Pi_N}\frac{v^{\Delta}_{S_{\pi, i}\cup \{j\}}}{v_{S_{\pi, i}}} \right) = \phi^{rel'}_i - \phi^{rel}_i
    \end{align}

    where $\Pi_{i \prec j}$ is a subset of $\Pi_N$ containing all permutations with $i$ preceding $j$. We consider only the subset $\Pi_{i \prec j}$ of $\Pi_N$ in the first equality since the quality of a model trained on the aggregated data can only change when containing party i. The first inequality is due to $v'_C \geq v_C$ for all $C \subseteq N$. The second equality can be understood by swapping the positions of $i$ and $j$ in permutation $\pi$ to obtain $\pi'$ such that $S_{\pi, j} \cup \{j\} = S_{\pi', i} \cup \{i\}$. The second inequality holds because there might be a permutation $\pi$ where $i \prec j$ in $\pi$ (i.e. $j \notin S_{\pi, i}$) and $v^{\Delta}_{S_{\pi, i} \cup \{i\}} \geq 0$. The last equality holds since $v^{\Delta}_{S_{\pi, i}} = v_{S_{\pi, i}}$ because $i\notin S_{\pi, i}$.

    Using this, we can finally proof F4: Let $l \in N$ be the player with the largest Shapley value $\phi'_l$ based on $v'$. Thus we have $\phi'_i \leq \phi'_l$ and $\phi'_i - \phi_i \geq \phi'_l - \phi_l$ due to party monotonicity when $l \neq i$. Together, it follows that $\phi_l \geq \phi_i$. Then
    \begin{align}
        \begin{split}
            \phi^{rel'}_i - \phi^{rel}_i & \geq \phi^{rel'}_l - \phi^{rel'}_l\\
            \phi^{rel}_l \left( \phi^{rel'}_i - \phi^{rel}_i \right) & \geq \phi^{rel}_i \left( \phi^{rel'}_l - \phi^{rel}_l \right)\\
            \phi^{rel}_l \times \phi^{rel'}_i & \geq \phi^{rel}_i \times \phi^{rel'}_l\\
            \frac{\phi^{rel'}_i}{\phi^{rel'}_l} & \geq \frac{\phi^{rel}_i}{\phi^{rel}_l}
        \end{split}
    \end{align}
    Let $p \in N$ be the party with the largest Shapley value $\phi_p$ based on $v$. Then
    \begin{align}
            r'_i = \left( \frac{\phi^{rel'}_i}{\phi^{rel'}_l} \right)^{\rho} \times v'_N
            \geq \left( \frac{\phi^{rel'}_i}{\phi^{rel'}_l} \right)^{\rho} \times v_N
            \geq \left( \frac{\phi^{rel}_i}{\phi^{rel}_l} \right)^{\rho} \times v_N
            \geq \left( \frac{\phi^{rel}_i}{\phi^{rel}_p} \right)^{\rho} \times v_N
            = r_i
    \end{align}
    For equality $r'_i \geq r_i$ to possibly hold, it has to be the case that $\phi^{rel}_i = \phi^{rel}_l = \phi^{rel}_p$ and $v'_N = v_N$. This implies $r'_i = v'_N = v_N = r_i$, so the premise $v'_N > r_i$ of F4 is not satisfied.
\end{proof}

\subsection{Proof of Individual Rationality (R4) and Stability}

Since our reward scheme follows the same structure as the additive Shapley reward scheme by Sim and colleagues, it suffers the same limitations regarding individual rationality and group stability. However, these can be overcome in the exact same way as done by Sim and colleagues. To achieve individual rationality, we have to ensure that 

\begin{align}
    \rho \leq \rho_r = \min_{i \in N}\frac{\log \left( \frac{v_i}{v_N} \right)}{\log \left( \frac{\phi^{rel}_i}{\phi^*} \right)}.
\end{align}

The stability of the grand coalition is defined by Sim and collegues as follows:\\
\\
\begin{tabularx}{\linewidth}{l X}
        R6 & \textbf{Stability of the Grand Coalition:} The grand coalition $N$ is stable if for every coalition $C \subset N$, the value of the model reward received by the party with the largest Shapely value is at least $v_C$:
        \begin{align}
            \forall C \subseteq N,  \forall i \in C: \phi_i = \max_{j \in C}\phi_i \Rightarrow v_C \geq r_i
        \end{align}
\end{tabularx}

Just as with the additive scaled Shapley value, this can be achieved by setting

\begin{align}
    \rho \leq \rho_s = \min_{j \in N}\frac{\log \left( \frac{v_{C_j}}{v_N} \right)}{\log \left( \frac{\phi^{rel}_j}{\phi^*} \right)}
\end{align}

Where for any $i \in N$, $C_i := \{ j | \phi_j \leq \phi_i  \}$, that is $C_i$ includes all agents whose Shapley value is at most $\phi_i$.

Together this gives us the following two theorems:

\begin{theorem}\label{theoR4}
    For all $\rho \leq \rho_r = \min_{i \in N}\frac{\log \left( \frac{v_i}{v_N} \right)}{\log \left( \frac{\phi^{rel}_i}{\phi^*} \right)}$, the rewards $\left( r_i \right)_{i \in N}$ satisfy individual rationality (R4).
\end{theorem}

\begin{theorem}\label{theoStab}
    For all $\rho \leq \rho_s = \min_{j \in N}\frac{\log \left( \frac{v_{C_j}}{v_N} \right)}{\log \left( \frac{\phi^{rel}_j}{\phi^*} \right)}$, the rewards $\left( r_i \right)_{i \in N}$ satisfy individual rationality (R4).
\end{theorem}

The proofs for theorems \ref{theoR4} and \ref{theoStab} are identical to the proofs for individual rationality and stability for the additive scaled Shapley value by Sim and colleagues \cite{sim2020collaborative}.

%%%%%%%%%%%%%%%%%%%%%%%%%%%%%%%%%%%%%%%%%%%%%%%%%%%%%%%%%%%%%%%%%%%%%%%%

\section{Comparison to Additive Shapley Rewards}

To highlight the practical differences and advantages of our ratio-based Shapley value, we compare it to the additive Shapley value used in the reward mechanism by Sim and colleagues \cite{sim2020collaborative}. Both approaches follow the same model-reward allocation scheme: Given a coalition value $v_C$, each participant receives a share of the model reward scaled relative to their contribution. To differentiate, we will from now on denote rewards according to our ratio-based reward scheme by $R_i$ and rewards according to the additive Shapley scheme introduced by Sim and colleagues by $A_i$. Thus::

\begin{align}
    R_i = \left( \frac{\phi^{rel}_i}{\phi^{rel, *}} \right) \times v_C; \quad  A_i = \left( \frac{\phi^{add}_i}{\phi^{add, *}} \right) \times v_C,
\end{align}

where $\phi^{rel}_i$ is the valuation of player $i$ according to our ratio based reward scheme and  $\phi^{rel, *} = \max_{j \in C} \phi^{rel}_j$, while  $\phi^{add}_i$ is the valuation of player $i$ according to the additive Shapley scheme with  $\phi^{add, *} = \max_{j \in C} \phi^{add}_j$. Furthermore,  $\rho \in (0, 1)$ modulates the strength of reward differentiation.

For a simple synthetic evaluation and comparison method,  we created 7 agents $\{1, ..., 7\}$. For each agent we set the standalone value as $v_i = \sqrt{i}$. Thus, agent $1$ has a value of 1 and agent 9 a value of 3. For any coalition, we then define its value as 
\begin{align}
    v_C = \sqrt{\sum_{i \in C}i}.
\end{align}

\begin{figure*}
	\centering
	\includegraphics[width=.9\textwidth]{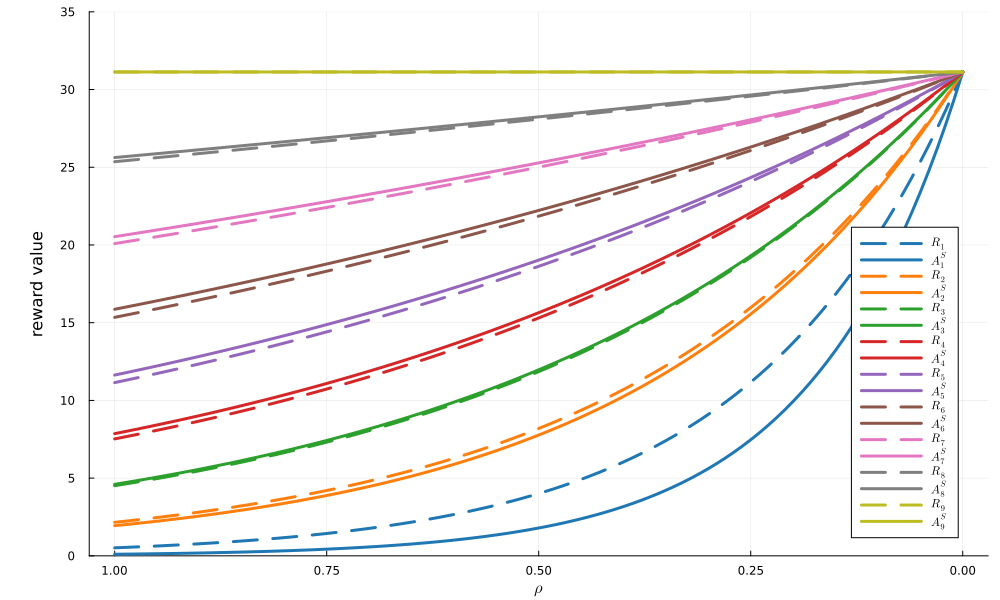}
	\caption{Rewards $R_i$ (solid lines) and $A_i$ (dashed lines) achieved by each agent $i$ within the grand coalition $N$ for different values of $\rho$.}
	\label{FIG:comp1}
\end{figure*}

The is the way our method emphasizes relative improvement over absolute marginal gains. Under the additive Shapley value, players receive credit proportional to the average additive increase their data brings across all coalitions. In contrast, the ratio-based value rewards players for improving the performance of coalitions proportionally, capturing scenarios where a small amount of data leads to large relative improvements, even if the absolute gain is modest. Naturally, the player with the largest additive marginal contributions is also the player with the largest relative contributions, just under both schemes the same player receives the maximal reward.

Notably, although the two formulations nearly converge in the extreme cases ($\rho = 0$ or $\rho = 1$), the reward curves differ significantly in the intermediate region. For example, certain players who are undervalued under the additive Shapley (due to small marginal contributions in high-value coalitions) receive noticeably higher rewards under our method. This is because they offer disproportionately high relative value in weaker coalitions. This effect is clearly visible for lower-ranked players (blue and orange) in Figure \ref{FIG:comp1}, whose rewards drop off more slowly in our method compared to the additive one. The reverse happens for players with large marginal contributions. While the receive significantly larger rewards, both under our and under the additive Shapley reward scheme, under our scheme their rewards are sightly smaller, Since their relative contributions are not quite as significant. Note, that such a change could not be brought about by simply shifting the $\rho$ parameter for the additive Shapley reward scheme. This would only increase or decrease all rewards, but could never decrease higher rewards while increasing lower rewards simultaneously.

This distinction has several useful consequences. It is more responsive to underrepresented data, players who contribute to low-performing coalitions and boosting their value significantly are rewarded more fairly. This can be crucial in settings with heterogeneous or unbalanced data distributions. In small or under-performing groups, a player who brings even modest data can produce significant relative improvements. The additive Shapley may undervalue such contributions, while the ratio-based value highlights them more appropriately. When collaborators have unequal resources or expertise, additive rewards may consistently favor the most resource-rich participants. Ratio-based rewards offer a fairness mechanism that can still incentivize smaller players who make relatively impactful contributions. Our reward scheme is also better suited to redundancy-prone domains: In contexts like sensor networks, health data sharing, or document digitization, where contributors often provide overlapping or similar data, additive rewards may overvalue already well-represented data. The ratio-based value down weights such redundancy by evaluating improvements relative to existing coalition strength.

Importantly, our formulation preserves the desirable incentive and fairness properties that underpin the scheme from Sim and collegues, such as individual rationality, feasibility, and symmetry, while offering an alternate lens through which to interpret and allocate value. This makes the ratio-based Shapley value a  valuable alternative in incentive-aware collaborative learning systems where the nature of value is better captured by proportional rather than additive gains.

Since the mathematical structure of the $\rho$-scaling remains unchanged, the same analytical tools, bounds, and stability analyses developed for additive Shapley rewards apply directly here with only the valuation function modified.

%%%%%%%%%%%%%%%%%%%%%%%%%%%%%%%%%%%%%%%%%%%%%%%%%%%%%%%%%%%%%%%%%%%%%%%%

\section{Conclusion}

In this paper, we have extended the framework of incentive-aware model rewards introduced by Sim and colleagues \cite{sim2020collaborative} by proposing a new cooperative game-theoretic solution concept: The ratio-based Shapley value. Whereas the original framework relies on the classical additive Shapley value to assess data contributions and rewarding players in proportion to their marginal improvement in coalition value, our approach shifts the focus to proportional gains, rewarding participants according to the relative improvement they bring to any given coalition.

This alternative is motivated by scenarios where relative performance matters more than absolute value, such as early-stage modeling, settings with highly heterogeneous data, or domains with strong redundancy effects. We formally defined the ratio-based Shapley value, embedded it in the $\rho$-scaled reward mechanism of Sim and colleagues, and proved that it satisfies all the core rationality and fairness axioms used in their framework: Feasibility, non-negativity, individual rationality (under $\rho$-scaling), symmetry, uselessness, strict desirability, and strict monotonicity.

The key theoretical insight of this paper is that the Shapley-based mechanism proposed by Sim and colleagues is not unique in satisfying these axioms. Our ratio-based value demonstrates that fundamentally different interpretations of contribution, additive versus multiplicative can yield reward functions that are both distinct in behavior and equally valid in terms of axiomatic justification. This shows that the current axiomatic framework allows for a non-trivial space of compatible mechanisms. In other words, the axioms do not uniquely determine the Shapley value as they do in classical cooperative game theory.

This observation leads to a fundamental open problem: How can we characterize the full set of reward mechanisms that satisfy the incentive-aware axioms? A formal representation theorem, describing all functions that fulfill the given axioms, would provide clarity about the boundaries of the framework and guide the selection of mechanisms in different application contexts.

Alternatively, if one seeks to recover uniqueness or constrain the space of acceptable solutions, then the current axiom set must be refined or extended. For instance, one might introduce new axioms that explicitly favor additive fairness, or conversely, ones that encode proportional equity or context-sensitivity. Such axioms could help disambiguate between competing notions of fairness in replicable reward settings and lead to more tailored mechanisms for specific domains.

Beyond its theoretical implications, our result has practical relevance: It broadens the toolbox for designing collaborative learning protocols, especially in environments where model access is the reward and fairness must account for nonlinear, context-dependent contributions. Future work may explore hybrid reward schemes that interpolate between additive and ratio-based values, investigate empirical behaviors of these schemes on real-world collaborative learning datasets, or extend the axiomatic framework to account for additional structural assumptions, such as coalition costs, communication constraints, or trust.

In sum, this paper contributes a concrete alternative to an increasingly important line of research in cooperative AI and collaborative machine learning. By highlighting the non-uniqueness of current fairness axioms and providing a viable new mechanism, we open the door to a richer understanding of what fairness means in collaborative data-driven systems.

%%%%%%%%%%%%%%%%%%%%%%%%%%%%%%%%%%%%%%%%%%%%%%%%%%%%%%%%%%%%%%%%%%%%%%%%

%%% Use this command to include your bibliography file.
%\bibliographystyle{splncs04}
 %\bibliography{refs}
%\printbibliography

%%%%%%%%%%%%%%%%%%%%%%%%%%%%%%%%%%%%%%%%%%%%%%%%%%%%%%%%%%%%%%%%%%%%%%
%%%%%%%%%%%%%%%%%%%%%%%%%%%%%%%%%%%%%%%%%%%%%%%%%%%%%%%%%%%%%%%%%%%%%%

\end{document}